\documentclass[11pt,letterpaper]{article}

\pdfoutput=1

\usepackage{amsmath}
\usepackage{amssymb}
\usepackage{amsthm}
\usepackage[in]{fullpage}
\usepackage{forloop}
\usepackage{bm}
\usepackage{bbm}
\usepackage{xspace}
\usepackage[l2tabu, orthodox]{nag}
\usepackage{nth}
\usepackage{color}
\usepackage{array}
\usepackage{paralist}
\usepackage{caption}
\usepackage{flafter}
\usepackage{tikz}
\usepackage{xifthen}
\usepackage{url}
\usepackage{times}

\usepackage{authblk}

\DeclareMathOperator*{\argmax}{arg\,max}

\newtheorem{theorem}{Theorem}
\newtheorem{lemma}[theorem]{Lemma}
\newtheorem{observation}[theorem]{Observation}
\theoremstyle{remark}
\newtheorem{remark}[theorem]{Remark}

\begin{document}
\title{{Online Graph Matching Problems\\with a Worst-Case Reassignment Budget\footnote{Part of this research was conducted while S. Lee was at Yonsei University. This work was supported by the National Research Foundation of Korea (NRF) grant funded by the Korea government (MSIT) (No. NRF-2019R1C1C1008934). This work was supported by the National Research Foundation of Korea (NRF) grant funded by the Korea government (MSIT) (No. NRF-2016R1C1B1012910).
}}}

\author[1]{Yongho Shin}
\author[1, 2]{Kangsan Kim}
\author[3]{Seungmin Lee}
\author[1]{Hyung-Chan An\thanks{Corresponding author. Email: hyung-chan.an@yonsei.ac.kr}}
\affil[1]{Department of Computer Science, Yonsei University, South Korea}
\affil[2]{Devsisters Corp., South Korea}
\affil[3]{TmaxData Co., Ltd., South Korea}

\date{}

\maketitle

\setcounter{page}{0}
\maketitle
\thispagestyle{empty}

\begin{abstract}
In the \emph{online bipartite matching with reassignments} problem, an algorithm is initially given only one side of the vertex set of a bipartite graph; the vertices on the other side are revealed to the algorithm one by one, along with its incident edges. The algorithm is required to maintain a matching in the current graph, where the algorithm revises the matching after each vertex arrival by reassigning vertices. Bernstein, Holm, and Rotenberg showed that an online algorithm can maintain a matching of maximum cardinality by performing \emph{amortized} $O(\log^2 n)$ reassignments per arrival.

In this paper, we propose to consider the general question of \emph{how requiring a \emph{non-amortized} hard budget $k$ on the number of reassignments affects the algorithms' performances,} under various models from the literature.

We show that a simple, widely-used algorithm is a best-possible deterministic algorithm for all these models. For the unweighted maximum-cardinality problem, the algorithm is a $(1-\frac{2}{k+2})$-competitive algorithm, which is the best possible for a deterministic algorithm both under vertex arrivals and edge arrivals. Applied to the \emph{load balancing problem}, this yields a bifactor online algorithm. For the weighted problem, which is traditionally studied assuming the triangle inequality, we show that the power of reassignment allows us to lift this assumption and the algorithm becomes a $\frac{1}{2}$-competitive algorithm for $k=4$, improving upon the $\frac{1}{3}$ of the previous algorithm without reassignments. We show that this also is a best-possible deterministic algorithm.
\end{abstract}


\newpage

\section{Introduction}\label{sec:intro}

%%%%%%%%%%%%%%%%%%%%%%%%%%%%
% Main
%%%%%%%%%%%%%%%%%%%%%%%%%%%%
Suppose that we have a set of servers, and that clients to be assigned to these servers arrive one by one.
Each time a client arrives, the subset of servers to which the client may be assigned is revealed.
In this setting, the \emph{online bipartite matching problem} asks us to \emph{irrevocably commit} whether the client will be assigned or not and, if assigned, to which server.
This commitment needs to be made immediately after each arrival, and the objective is to maximize the number of assigned clients so that every server is assigned at most one client.
This problem was first considered by the seminal work of Karp, Vazirani, and Vazirani~\cite{DBLP:conf/stoc/KarpVV90}, which gives the best-possible randomized $(1 - \frac{1}{e})$-competitive algorithm.
Their analysis was later improved (with the same competitive ratio, of course) by Goel \& Mehta~\cite{DBLP:conf/soda/GoelM08}, Birnbaum \& Mathieu~\cite{DBLP:journals/sigact/BirnbaumM08}, and Devanur, Jain \& Kleinberg~\cite{DBLP:conf/soda/DevanurJK13}.

In that problem, once the algorithm commits itself to match a vertex (i.e., a server or a client), it ensures that the vertex remains matched from now onward.
However, the choice of a particular \emph{edge} in the matching may not be as important in some applications.
Thus, one of natural goals would be to design an online algorithm that still makes irrevocable commitments to vertices that they will be matched, but performs (a minimum number of) reassignments.
Introduced by Grove, Kao, Krishnan, and Vitter~\cite{DBLP:conf/wads/GroveKKV95}, this question was considered also by Bernstein, Holm, and Rotenberg~\cite{DBLP:journals/jacm/BernsteinHR19} who showed an online algorithm can indeed maintain a maximum-cardinality matching by performing \emph{amortized} $O(\log^2 n)$ reassignments per arrival, i.e., $O(n \log^2 n)$ reassignments in total.

However, what if we require a \emph{worst-case} bound?
In this paper, we propose to consider, under a variety of models from the literature, the general question of \emph{how the algorithms' performances are affected when a hard (non-amortized) budget $k$ is imposed on the number of reassignments per arrival}.
We will demonstrate that a simple, widely-used algorithm is a best-possible deterministic algorithm under all these models.
The algorithm, which was also used to give the amortized $O(\log^2 n)$ bound, is indeed simple to describe here:
\begin{quote}
Find a shortest (or most profitable, if the model is weighted) augmenting path with at most $k$ vertices, starting from the newly-arrived one; augment the matching along the path.
\end{quote}

Let us formally describe the classic model of \emph{maximum-cardinality online bipartite matching under vertex arrivals}, augmented with a \emph{hard budget} on the number of reassignments.
For a bipartite graph $G = (L \cup R, E)$, the algorithm is initially given $R$ and the hard budget $k$.
In the next $|L|$ timesteps, the vertices in $L$ are revealed one by one along with its incident edges; the algorithm outputs at most $k$ (re)assignments that give a new matching at each timestep.
If some vertices in $L \cup R$ become matched as a result, this means that the algorithm is now committed to keep them matched.
The new matching, of course, must respect all the previous commitments the algorithm has made so far.
When $k = 2$, i.e., an algorithm is allowed to perform only assignments but no reassignments (note that adding an edge to the matching already involves two assignment operations on its two endpoints), the given algorithm is $\frac{1}{2}$-competitive~\cite{DBLP:conf/stoc/KarpVV90}.
In Section~\ref{subsec:analysis-cardinality-vertex}, we show the following generalization, a tradeoff between $k$ and the competitive ratio.
\begin{theorem} \label{analysis-cardinality:vertex-arrival}
The given algorithm is a $(1 - \frac{2}{k + 2})$-competitive algorithm for the maximum-cardinality online bipartite matching problem under vertex arrivals with a hard budget of $k$ on the number of reassignments.
\end{theorem}
This is proven by first showing that the algorithm has an invariant that the current graph, with respect to the current matching, does not have an augmenting path of length at most $k-1$. The proof of this invariant crucially relies on the fact that the algorithm chooses shortest augmenting paths---in fact, an arbitrary choice of augmenting paths may allow short augmenting paths to exist when $k> 4$: see Observation~\ref{obs:appendix} in Appendix~\ref{sec:app-defer}. Once we have this invariant though, the desired conclusion follows from the fact that such a matching is a good approximation to the maximum one.

Theorem~\ref{analysis-cardinality:vertex-arrival} is complemented with the following lower bound (see Section~\ref{subsec:analysis-cardinality-hard}).
\begin{theorem} \label{analysis-cardinality-hard:lem-hard}
For all constant $\epsilon > 0$, no deterministic algorithm is $(1 - \frac{2}{k + 2} + \epsilon)$-competitive for the maximum-cardinality online bipartite matching problem under vertex arrivals with a hard budget of $k$ on the number of reassignments.
\end{theorem}
 
When the graph is weighted, it is easy to see that no online algorithm can have a positive constant competitive ratio.
As such, Kalyanasundaram and Pruths~\cite{DBLP:journals/jal/KalyanasundaramP93} study the model where the algorithm is required to maintain a left-perfect matching and the edge costs are assumed to form a metric, to give a $\frac{1}{3}$-competitive algorithm when no reassignments are allowed (i.e., $k = 2$).
We remark that we can show both the metric assumption and the left-perfect requirement are necessary for this problem to admit a constant competitive ratio.\footnote{This is true even when we allow the ``additive error'' of the competitive ratio to depend on the maximum edge weight.
When we only require left-perfect matching, consider the following oblivious adversary. The adversary constructs a complete bipartite graph with $|L|=|R|=n$, where $n$ (and $m$) are to be chosen later. Let $L'$ be the last $m$ vertices of $L$, according to the arrival order. Let $R'\subseteq R$ be a set of $m$ vertices, sampled uniformly at random. Every edge between $L'$ and $R'$ has weight $1$, and all other edges have weight $0$.
\newline
When we only assume metric weights, consider the following oblivious adversary.
Let $L := \{ u_{ij} \mid 1\leq i\leq N, 1\leq j\leq M \}$ and $R := \{ v_i \mid 0\leq i\leq N \}$, where $N$ and $M$ are to be chosen later.
The edge weights of the complete bipartite graph is defined as $w( u_{ij}, v_{i^\prime} ) = r^{j}$, where $r$ also is to be chosen later.
The adversary reveals the vertices in $L$ in a non-decreasing order of index $j$. \newline
For two constants $\rho\leq 1$ and $c\geq 0$, in order to see that no randomized online algorithm can output a solution of expected value at least $\rho\cdot \mathsf{OPT} - c$, we choose $m>\frac{c}{\rho}$ and $n > \frac{m^2}{\rho m - c}$ for the first case; we choose $r = M = \lceil \frac{2}{\rho} \rceil$ and $N > c(c + r)$ for the second case.
Note that only $N$ depends on $c$ in the second case.}
In Section~\ref{subsec:analysis-weight-comp}, we consider the case where reassignments are allowed, and show that the power of reassignments enables us to remove the metric assumption and obtain a better competitive ratio.
\begin{theorem} \label{analysis-weight-comp:main}
The given algorithm is a $\frac{1}{2}$-competitive algorithm for the maximum-weight online bipartite left-perfect matching problem with general weights under vertex arrivals with a hard budget $k = 4$ on the number of reassignments.
\end{theorem}
This again is complemented with the following lower bound, which holds even for $\{0, 1\}$-weighted graphs (see Section~\ref{subsec:analysis-weight-hard}).
\begin{theorem} \label{analysis-weight-hard:lem-hard}
For all constant $\epsilon > 0$, no deterministic algorithm is $(\frac{1}{2} + \epsilon)$-competitive for the maximum-weight online bipartite left-perfect matching problem under vertex arrivals with $k = 4$, even when all edge weights are either $0$ or $1$.
\end{theorem}

Returning to the unweighted realm, nonbipartite graphs have also been the target of interest~\cite{DBLP:conf/stoc/0002KTWZZ18, DBLP:conf/soda/0002PTTWZ19, DBLP:conf/focs/GamlathKMSW19}.
Recently, Gamlath, Kapralov, Maggiori, Svensson, and Wajc~\cite{DBLP:conf/focs/GamlathKMSW19} gave a $(\frac{1}{2} + \epsilon)$-competitive algorithm for some constant $\epsilon > 0$.
Our algorithm applies to this model as well, yielding a best-possible deterministic algorithm.
We will not directly prove this, though, since it is implied by Theorem~\ref{analysis-cardinality:vertex-arrival} and the proof of Theorem~\ref{analysis-cardinality-hard:lem-hard}.
Another widely studied setting that further generalizes the nonbipartite case is the edge-arrival model.
This model is traditionally defined by considering edges as the ``online object'' to which irrevocable commitments need to be made, and therefore the problem definition becomes slightly less natural in conjunction with reassignments.
We will however define an edge-arrival model partially because it still serves as a generalization of other models.
In this model, it is edges that arrive at each timestep, but the algorithm still makes commitments to vertices that they will be matched.
The given algorithm---if we find an augmenting path ``containing the newly-arrived edge'' rather than ``starting from the newly-arrived one''---is a competitive algorithm (see Section~\ref{subsec:analysis-cardinality-vertex}):
\begin{theorem} \label{analysis-cardinality:edge-arrival}
The given algorithm is a $(1 - \frac{2}{k + 2})$-competitive algorithm for the maximum-cardinality online matching problem under edge arrivals with a hard budget of $k$ on the number of reassignments.
\end{theorem}
We omit the proof of the following lower bound as it is implied by Theorem~\ref{analysis-cardinality-hard:edge-adapt-online}.
\begin{theorem}
For all constant $\epsilon > 0$, no deterministic algorithm is $(1 - \frac{2}{k + 2} + \epsilon)$-competitive for the maximum-cardinality online matching problem under edge arrivals with a hard budget of $k$ on the number of reassignments.
\end{theorem}

%%%%%%%%%%%%%%%%%%%%%%%%%%%%
% Future work
%%%%%%%%%%%%%%%%%%%%%%%%%%%%
\subparagraph*{Other Results and Future Directions.}
We believe one natural question is how online problems in general admit better algorithms when we allow a nonzero budget on the number of modifications to the solution at each timestep.
For a small example, our results bear implications on the \emph{load balancing problem}, as was the case for other previous results on the online matching problem~\cite{DBLP:conf/soda/GuptaKS14, DBLP:conf/innovations/BernsteinKPPS17, DBLP:journals/jacm/BernsteinHR19}.
In this problem, servers can be assigned multiple clients, but the objective is to minimize the maximum number of clients assigned to a single server.
Gupta, Kumar, and Stein \cite{DBLP:conf/soda/GuptaKS14} give a $8$-competitive algorithm for this problem with \emph{amortized} $O(1)$ reassignments per arrival, and Bernstein et al.~\cite{DBLP:journals/jacm/BernsteinHR19} give a $(1 + \epsilon)$-competitive algorithm with \emph{worst-case} $O \left( \frac{1}{\epsilon} \log n \right)$ reassignments per arrival.
Theorem~\ref{analysis-cardinality:vertex-arrival} leads to the following bifactor online algorithm giving an explicit tradeoff between the hard budget and the ratio of served clients. Its proof follows from the standard guess-and-double framework (see, e.g.,~\cite{DBLP:conf/soda/GuptaKS14,DBLP:conf/innovations/BernsteinKPPS17}) and is omitted here.
\begin{theorem}
For the load balancing problem with a hard budget $k$ on the number of reassignments, there exists a $4$-competitive algorithm that assigns at least $(1-\frac{2}{k+2})$ fraction of the clients.
\end{theorem}

Another important question is on the power of randomness in the problems considered in this paper.
While we show that the given algorithm yields the best-possible competitive ratios of a deterministic algorithm for these problems, it remains open whether randomization helps.
In this line, we conclude with the following theorem that shows randomization does \emph{not} help under the edge-arrival model against adaptive online adversaries. (See Section~\ref{subsec:analysis-cardinality-hard}.)
\begin{theorem} \label{analysis-cardinality-hard:edge-adapt-online}
For all constant $\epsilon > 0$, no randomized algorithm is $(1 - \frac{2}{k + 2} + \epsilon)$-competitive against an adaptive online adversary for the maximum-cardinality online matching problem under edge arrivals with a hard budget of $k$ on the number of reassignments.
\end{theorem}

%%%%%%%%%%%%%%%%%%%%%%%%%%%%
% Related work
%%%%%%%%%%%%%%%%%%%%%%%%%%%%
\subparagraph*{Related Work.}
Grove et al.~\cite{DBLP:conf/wads/GroveKKV95} first introduced the notion of \emph{reassignments} to the online bipartite matching problem and showed that, for a special case where every vertex in $L$ has degree at most two, an online algorithm can maintain a left-perfect matching by performing amortized $O(\log n)$ reassignments per arrival. Chaudhuri, Daskalakis, Kleinberg, and Lin~\cite{DBLP:conf/infocom/ChaudhuriDKL09} considered various settings including the case where the graph is a forest in particular, for which they give an algorithm with amortized $O(\log n)$ reassignments. Bosek, Leniowski, Sankowski, and Zych~\cite{DBLP:conf/focs/BosekLSZ14} considers the general version of the problem and obtain an algorithm with worst-case $O(\sqrt n)$ reassignments. Bosek, Leniowski, Sankowski, and Zych-Pawlewicz~\cite{DBLP:conf/latin/BosekLSZ18} show that taking shortest augmenting paths results in amortized $O(\log n)$ reassignments for trees.

There was another proposal on relaxing irrevocability, which is to specify a limit on the number of times each edge can enter or leave the matching. This problem is sometimes called \emph{online matching with recourse}.
Avitabile, Mathieu, and Parkinson~\cite{DBLP:journals/ipl/AvitabileMP13} considered a general class of online packing and covering problems and showed a $(1 - \epsilon)$-competitive algorithm for packing problems with the limit of $\widetilde{O}\left( \frac{1}{\epsilon} \right)$.
Epstein, Levin, Segev, and Weimann~\cite{DBLP:journals/iandc/EpsteinLSW18} considered the case where only edge departures are irrevocable and gave 0.186-competitive algorithms for the weighted case.
Boyar, Favrholdt, Kotrb\v{c}\'{\i}k, and Larsen~\cite{DBLP:conf/wads/BoyarFKL17} gave a tight $\frac{2}{3}$-competitive algorithm for the case where the decision on an edge can be revoked at most once. Although their model significantly differs from our proposed model, we note that their analysis, with hindsight, would have been useful in proving some special cases ($k=4$) of Theorems~\ref{analysis-cardinality:vertex-arrival}~and~\ref{analysis-cardinality:edge-arrival}.
Angelopoulos, D\"{u}rr, and Jin~\cite{DBLP:conf/mfcs/0001DJ18} considered the case where the decision can be revoked at most $k-1$ times and gave a $ \left( 1 - O\left( \frac{1}{\sqrt{k}} \right) \right)$-competitive algorithm.

The maximum-cardinality online bipartite matching problem has also been extensively studied for decades in a variety of settings and applications, including fully online models~\cite{DBLP:conf/icalp/WangW15} and stochastic settings~\cite{DBLP:conf/focs/FeldmanMMM09, DBLP:conf/soda/ManshadiGS11, DBLP:conf/wine/HaeuplerMZ11, DBLP:conf/soda/MehtaWZ15}.
This is the case for the weighted problem as well: subsequent to the introduction of the problem by Kalyanasundaram~\& Pruhs~\cite{DBLP:journals/jal/KalyanasundaramP93} and Khuller, Mitchell, \& Vazirani~\cite{DBLP:conf/icalp/KhullerMV91}, many results were obtained under the \emph{free disposal model}, where a newly-arrived vertex can take away the match of another vertex~\cite{ DBLP:conf/wine/FeldmanKMMP09, DBLP:conf/wine/KorulaMZ13, DBLP:journals/corr/Zadimoghaddam17, DBLP:journals/corr/abs-1910-03287}.
The weighted problem was also considered in vertex-weighted models~\cite{DBLP:conf/soda/AggarwalGKM11, DBLP:conf/icalp/0002TWZ18}, stochastic settings~\cite{DBLP:conf/wine/HaeuplerMZ11, DBLP:conf/esa/BrubachSSX16, DBLP:conf/soda/GamlathKS19}, random arrivals~\cite{DBLP:conf/esa/KesselheimRTV13}, and the AdWords problem~\cite{DBLP:journals/jacm/MehtaSVV07} for example. 
The minimization version also has been studied under various settings~\cite{DBLP:conf/icalp/KhullerMV91, DBLP:journals/jal/KalyanasundaramP93, DBLP:conf/soda/MeyersonNP06, DBLP:journals/algorithmica/BansalBGN14, DBLP:conf/approx/Raghvendra16, DBLP:conf/icalp/GuptaGPW19}.
\section{Our Models and Algorithm} \label{sec:modalg}

In this section, we present our models and the algorithm. While the models were already introduced in Section~\ref{sec:intro}, we describe them here again for the sake of completeness. We begin with recalling a few standard definitions: given a graph $G=(V,E)$ and its matching $M\subseteq E$, we say a vertex is \emph{matched} if it is an endpoint of some edge in $M$, and \emph{exposed} otherwise. Given two sets $A$ and $B$, $A\triangle B$ denotes their symmetric difference: $A\triangle B:=(A\setminus B)\cup (B\setminus A)$. For a path $P$ in $G$, we will slightly abuse the notation and treat $P$ as a set of edges when it is clear from the context.

\subparagraph*{Models.}
To define the \emph{maximum-cardinality online bipartite matching problem under vertex arrivals with a hard budget}, consider a bipartite graph $G=(L\cup R, E)$. The algorithm is initially given just $R$ and a hard budget $k$. The rest of the graph is gradually revealed to the algorithm over $|L|$ timesteps: at each timestep, a vertex in $L$ ``arrives'', at which point all its incident edges are revealed to the algorithm. The algorithm is required to output a bipartite matching in the current graph at the end of the timestep. The new matching must be obtained by performing at most $k$ (re)assignments starting from the previous one; if the algorithm matches a vertex at some point, it must remain matched.

A (re)assignment is defined as (re)assigning a vertex to another vertex: e.g., if the matching changes from $\{(x,y)\}$ to $\{(w,x),(y,z)\}$, this counts as four (re)assignments, each on $w$, $x$, $y$, and $z$. The ``previous'' matching of the first timestep is defined as the empty set.

In the \emph{maximum-cardinality online matching problem under edge arrivals with a hard budget}, the graph $G=(V,E)$ is not necessarily  bipartite. The algorithm is initially given $V$ and $k$. The edges of the graph are revealed to the algorithm one by one, over $|E|$ timesteps. The rest of the model is the same as the bipartite vertex-arrival case: the algorithm outputs a matching that is obtained by performing at most $k$ (re)assignments. If the algorithm matches a vertex at some point, it must remain matched.

Finally, in the \emph{maximum-weight online bipartite left-perfect matching problem under vertex arrivals with a hard budget}, we have a weighted complete bipartite graph $G=(L\cup R, E)$. The algorithm is initially given $R$ and $k$, and the vertices in $L$ arrives one by one along with all its incident edges and their weights. The algorithm then outputs a bipartite \emph{left-perfect} matching in the current graph, which must be obtained by performing at most $k$ (re)assignments. Again, if the algorithm matches a vertex at some point, it must remain matched.

In all these problems, we can assume without loss of generality that the hard budget $k$ is even. The following observation justifies why.
\begin{observation}
If an algorithm for the above problems outputs a matching $M_1$ at a timestep and $M_2$ at the next one, the minimum number of (re)assignments required to obtain $M_2$ from $M_1$ is even.
\end{observation}
\begin{proof}
Consider the graph $H=(V,M_1\triangle M_2)$. Let $V'$ be the set of vertices with nonzero degrees in this graph. Note that we can obtain $M_2$ from $M_1$ by performing exactly one (re)assignment on each vertex in $V'$, and this is the minimum number of reassignments.

Graph $H$ is a union of vertex-disjoint alternating paths and alternating cycles. Every alternating cycle has an even number of vertices. Every alternating path must be an augmenting path because of the requirement that every matched vertex must remain so; thus, every path must also contain an even number of vertices.
\end{proof}

\subparagraph*{Algorithm.}
When an object (vertex or edge) arrives, we find a shortest (if the model is unweighted) or most profitable (if weighted) one among all augmenting paths that contain the object.\footnote{It is easy to show that, in vertex-arrival models, the path will not just contain the new vertex but start from it.} If such a path does not exist or is longer than $k-1$, do nothing; otherwise, augment the matching along the path.

\newcommand{\scriptP}{\mathcal{P}}
\newcommand{\scriptQ}{\mathcal{Q}}
\newcommand{\scriptD}{\mathcal{D}}
\newcommand{\scriptX}{\mathcal{X}}
\newcommand{\scriptY}{\mathcal{Y}}

\section{Unweighted Case} \label{sec:analysis-cardinality}

\subsection{Analysis} \label{subsec:analysis-cardinality-vertex}

In this section, we present the analysis of the given algorithm under the unweighted models.

The proof of the following lemma can be found in many places including \cite{DBLP:journals/siamcomp/HopcroftK73}; we include a brief proof sketch for the sake of completeness.

\begin{lemma} [\cite{DBLP:journals/siamcomp/HopcroftK73}]
\label{prelim:lem-long-aug}
Let $G$ be a graph, $M$ be a matching in $G$, and $k$ be a positive even number.
If the length of every augmenting path of $G$ with respect to $M$ is at least $k + 1$,
$$
|M| \geq \left(1 -  \frac{2}{k + 2} \right) \mathsf{OPT},
$$
where $\mathsf{OPT}$ is the size of a maximum matching in $G$.
\end{lemma}
\begin{proof} [Proof sketch]
Let $M^\star$ be a maximum matching in $G$. Suppose $|M|\neq|M^\star|$.
Observe that $M \triangle M^\star$ is the union of vertex-disjoint alternating cycles and alternating paths; let $\mathcal{S}$ be the set of these cycles and paths.
We then have $\frac{|M|}{|M^\star|}\geq\min_{S\in\mathcal{S}}\frac{|M\cap S|}{|M^\star\cap S|}$. Every alternating cycle has the same number of edges from $M$ and $M^\star$; every augmenting path $S\in\mathcal{S}$ has at least $k+1$ edges and hence $\frac{|M\cap S|}{|M^\star\cap S|}\geq\frac{\lfloor \frac{k+1}{2}\rfloor}{\lceil \frac{k+1}{2}\rceil}=1-\frac{2}{k+2}$.
\end{proof}

The following lemma describes an operation we can perform on two augmenting paths, one of which is with respect to the matching that results from the augmentation using the other one. While it is tempting to imagine a simpler surgical operation satisfying the desired conditions, the correct version can be rather complicated when the graph is not bipartite. In favor of simplicity, we opted for the current version of the proof.

\newpage

 \begin{lemma} \label{keylemma:main}
Let $G = (V, E)$ be a graph (which is not necessarily bipartite), and let $M \subseteq E$ be a matching in $G$.
Let $P$ be an augmenting path in $G$ with respect to $M$, and let $Q$ be an augmenting path in $G$ with respect to $M \triangle P$.
Then, there exist two augmenting paths $X$ and $Y$ in $G$ with respect to $M$ such that
\begin{enumerate}
\item \label{keylemma:cond-disjoint} $X$ and $Y$ are edge-disjoint;
\item \label{keylemma:cond-contained} $X\cup Y \subseteq P\cup Q$; and
\item \label{keylemma:cond-feasible} the set of the four endpoints of $X$ and $Y$ equals the set of the four endpoints of $P$ and $Q$.
\end{enumerate}
\end{lemma}
\begin{proof}
Every vertex on $P$ is matched in $M\triangle P$; hence, the two endpoints of $Q$ are not on $P$.

Consider the graph $H:=(V,P\triangle Q)$. In this graph, exactly four vertices have odd degree: the endpoints of $P$ and of $Q$. All other vertices have even degree.
We claim that no vertex in $H$ has degree 3 or higher: suppose towards contradiction that a vertex $v$ has degree $3$ or higher. Since the endpoints of $Q$ are not on $P$, $v$ must be an internal vertex of $Q$. Since $v$ is on $P$, there exists some edge $e\in P\setminus M$ incident with $v$. Since $e$ is in the matching $M\triangle P$, $e$ must be in $Q$ as well and therefore $e\notin P\triangle Q$.

This shows that $H$ is a union of exactly two paths and some cycles, which are vertex-disjoint. Let $X$ and $Y$ be these two paths, and we can see that the three conditions of this lemma are satisfied. Since the endpoints of $Q$ are not on $P$, they must be exposed in $M$; it only remains to show that $X$ and $Y$ are alternating paths.

Note that, for every edge in $Q\setminus P$, the edge is in $M$ if and only if it is in $M\triangle P$. Let $v$ be a vertex of degree $2$ in $H$. If both of its incident edges are from $P$, exactly one of them is in $M$ and there is nothing to prove. This is the case when both are from $Q$ as well.
Otherwise, $v$ must be on both $P$ and $Q$. Since the endpoints of $Q$ are not on $P$, $v$ must be an internal vertex of $Q$; a similar argument to the above shows that there must exist some edge $e\in E$ incident with $v$ that belongs to both $P\setminus M$ and $Q\cap(M\triangle P)$. Thus, one of the two edges incident with $v$ in $P\triangle Q$ is in $P\cap M$, and the other is in $Q\setminus (M\triangle P)=Q\setminus M$.
\end{proof}

%%%%%%%%%%%%%%%%%%%%%%%%%%%%%%%%
% Analysis
%%%%%%%%%%%%%%%%%%%%%%%%%%%%%%%%
Now we present the competitive analysis for the bipartite vertex-arrival case.

\begingroup
\def\thetheorem{\ref{analysis-cardinality:vertex-arrival}}
\begin{theorem} [restated]
The given algorithm is a $(1 - \frac{2}{k + 2})$-competitive algorithm for the maximum-cardinality online bipartite matching problem under vertex arrivals with a hard budget of $k$ on the number of reassignments.
\end{theorem}
\addtocounter{theorem}{-1}
\endgroup

\begin{proof}
We will show that the length of every augmenting path is at least $k+1$ at the end of each timestep. The desired conclusion will follow from Lemma~\ref{prelim:lem-long-aug}.

Use induction on time $t$. We define the matching at the end of ``time 0'' as the empty set, and then the base case is trivial.
Now consider a timestep $t\geq 1$.
Let $G_{t-1}$ and $G_{t}$, respectively, be the graph at time $t-1$ and $t$. Let $M_{t-1}$ be the matching that is output by the algorithm at the end of time $t-1$, and let $u$ be the vertex that newly arrives at time $t$. From the induction hypothesis, every augmenting path in $G_{t-1}$ with respect to $M_{t-1}$ has length at least $k + 1$.

If the algorithm does nothing at time $t$, this is because the length of every augmenting path in $G_{t}$ starting from $u$ with respect to $M_{t-1}$ is at least $k + 1$. Therefore, if there exists some other augmenting path in $G_{t}$ of length no more than $k - 1$, this path does not contain $u$ and is an augmenting path also in $G_{t-1}$, contradicting the induction hypothesis.

Suppose that, at time $t$, the algorithm augments $M_{t - 1}$ along a shortest augmenting path $P$ in $G_{t}$.
Let $Q$ be an arbitrary augmenting path in $G_{t}$ with respect to $M_{t - 1} \triangle P$.
Then, by Lemma~\ref{keylemma:main}, there exist two augmenting paths $X$ and $Y$ in $G_{t}$ with respect to $M_{t - 1}$ satisfying the conditions of the lemma.
Without loss of generality, suppose that $u$ is an endpoint of $X$.
We then have $|X| \geq |P|$ since $X$ is an augmenting path in $G_{t}$ starting from $u$ with respect to $M_{t - 1}$.
Since $u$ is an endpoint of $X$, $Y$ does not contain $u$. Observe that $Y$ is an augmenting path in $G_{t - 1}$ with respect to $M_{t - 1}$ from Conditions~\ref{keylemma:cond-disjoint}~and~\ref{keylemma:cond-feasible} of Lemma~\ref{keylemma:main}; from the induction hypothesis, we have $|Y| \geq k + 1$.
By Condition~\ref{keylemma:cond-contained} of the lemma, we have $ |P| + |Q| \geq |X| + |Y| \geq |P| + k + 1$, implying  $|Q| \geq k + 1$.
\end{proof}

\begin{remark}
In the nonbipartite case, where the arrival of vertex $v$ reveals only the edges between $v$ and the other vertices that have already arrived, the given algorithm achieves the same competitive ratio, since neither Lemma~\ref{keylemma:main} nor the proof of Theorem~\ref{analysis-cardinality:vertex-arrival} assumes bipartiteness.
\end{remark}

A similar argument gives a competitive analysis under edge arrivals. We defer its proof to Appendix~\ref{sec:app-defer}.

\begingroup
\def\thetheorem{\ref{analysis-cardinality:edge-arrival}}
\begin{theorem} [restated]
The given algorithm is a $(1 - \frac{2}{k + 2})$-competitive algorithm for the maximum-cardinality online matching problem under edge arrivals with a hard budget of $k$ on the number of reassignments.
\end{theorem}
\addtocounter{theorem}{-1}
\endgroup

%%%%%%%%%%%%%%%%%%%%%%%%%%%%%%%%
% Lower bounds
%%%%%%%%%%%%%%%%%%%%%%%%%%%%%%%%
\subsection{Lower Bounds} \label{subsec:analysis-cardinality-hard}

In this section, we present the lower bounds under the unweighted models, starting with vertex-arrival models.

\begingroup
\def\thetheorem{\ref{analysis-cardinality-hard:lem-hard}}
\begin{theorem} [restated]
For all constant $\epsilon > 0$, no deterministic algorithm is $(1 - \frac{2}{k + 2} + \epsilon)$-competitive for the maximum-cardinality online bipartite matching problem under vertex arrivals with a hard budget of $k$ on the number of reassignments.
\end{theorem}
\addtocounter{theorem}{-1}
\endgroup

\begin{proof}

We will first present an adversary against which any deterministic algorithm outputs a matching of size strictly smaller than the optimum. We will then describe how we can modify this adversary to show the desired lower bound. Recall that oblivious adversaries are as powerful as adaptive offline adversaries against deterministic algorithms (see, e.g.,~\cite{DBLP:journals/algorithmica/Ben-DavidBKTW94}). As such, we assume an adaptive offline adversary in this proof.

Let $n := \frac{k}{2} + 1$. Let $G=(L\cup R,E)$ be a bipartite graph with $|L|=|R|=n$, whose vertices are named as follows: $L:=\{u_1,\ldots, u_n\}$ and $R:=\{v_1,\ldots, v_n\}$. The adversary gradually reveals this $G$ while maintaining the following invariant, possibly by renaming some vertices, right before the arrival of the $t$-th vertex:\begin{quote}
For all $1\leq i<t$, $u_i$ has already arrived and is adjacent to exactly two vertices, $v_i$ and $v_{i+1}$, and $(u_i,v_{i+1})$ is in the current matching.
\end{quote}
Note that the graph is therefore just a long path. Due to the renaming, the indices of $u_1,\ldots,u_{t-1}$ may not reflect the order of arrivals.

At time $t$ ($1\leq t<n$), the adversary reveals two edges $(u_i,v_i)$ and $(u_i,v_{i+1})$. The algorithm then has only three options: \textsf{(a)} do nothing; \textsf{(b)} add $(u_i,v_{i+1})$ to the matching; \textsf{(c)} add $(u_i,v_{i})$ to the matching. For the last option though, in order to honor the commitment we made to vertices, we need to take the augmenting path $\langle u_t,v_t,u_{t-1},\ldots,v_1\rangle$.

If the algorithm chooses Option~\textsf{(a)}, the adversary can immediately stop: observe that the graph has a matching of cardinality $t$, and the algorithm has already failed to output a maximum matching. If the algorithm chooses Option~\textsf{(b)}, the invariant remains satisfied, and we can proceed to the next timestep. If the algorithm chooses Option~\textsf{(c)}, we can maintain the invariant by renaming the vertices and proceed.

At time $t=n$, the adversary reveals only $(u_t,v_t)$. Option~\textsf{(b)} therefore does not exist, and Option~\textsf{(c)} cannot be taken since the augmenting path is too long. The algorithm is now forced to choose Option~\textsf{(a)} and output a suboptimal matching.

Now we prove the theorem. Suppose towards contradiction that there is an algorithm that always outputs a matching of size at least $(1 - \frac{2}{k + 2} + \epsilon) \cdot \mathsf{OPT} - c$ for some constant $\epsilon>0$ and $c\geq 0$, where $\mathsf{OPT}$ is the offline optimum of the graph the adversary produces.
Consider an adversary that ``simulates'' $N$ copies of the above adversary, where $N$ is a sufficiently large integer to be chosen later.
We will keep the graphs constructed by these adversaries mutually disjoint.
At each time, we choose an arbitrary copy of adversary whose current matching is left-perfect. The ``global'' adversary reveals the vertex that this copy would reveal.

Since each copy of adversary forces the matching to be suboptimal within $n$ timesteps, within $n\cdot N$ timesteps, every copy of adversary forces a suboptimal matching. By choosing $N > \frac{2c}{\epsilon (k + 2)}$, one can verify that the algorithm finds a matching of cardinality no greater than $\mathsf{OPT}-N< (1 - \frac{2}{k + 2} + \epsilon) \cdot \mathsf{OPT} - c$ since $\mathsf{OPT} \leq n\cdot N$.
\end{proof}
Now we consider the edge-arrival model against adaptive online adversaries. Note that this immediately implies the same lower bound against adaptive offline adversaries.

\newcommand{\E}{\mathbb{E}}

\begingroup
\def\thetheorem{\ref{analysis-cardinality-hard:edge-adapt-online}}
\begin{theorem} [restated]
For all constant $\epsilon > 0$, no randomized algorithm is $(1 - \frac{2}{k + 2} + \epsilon)$-competitive against an adaptive online adversary for the maximum-cardinality online matching problem under edge arrivals with a hard budget of $k$ on the number of reassignments.
\end{theorem}
\addtocounter{theorem}{-1}
\endgroup

\begin{proof}
Let $\mathsf{ADV}$ denote the cardinality of the online solution produced by the adversary. Suppose towards contradiction that there is a randomized online algorithm that outputs a matching whose expected cardinality is at least $(1 - \frac{2}{k + 2} + \epsilon) \cdot \E[\mathsf{ADV}] -c$ for some constant $c\geq 0$.

Now consider the following adversary. The adversary chooses the number of vertices as $(k+2)N$ for some sufficiently large integer $N$. In the first $\frac{kN}{2}$ timesteps, it reveals an edge that is disjoint from all of the previously revealed edges; the adversary does not include these edges in its solution. At the end of time~$\frac{kN}{2}$, the graph will contain $kN$ vertices of degree one and $2N$ of degree zero. Let $F$ be the set of edges revealed so far and $U$ be the set of the $2N$ vertices of zero degree. We will show later that the edges in $F$ are the only edges the algorithm can ever add to the matching, and once the algorithm adds an edge in $F$, the edge will never leave the matching.

The adversary then selects $\frac{k}{2}$ edges in the current matching output by the algorithm. (If the algorithm chose fewer than $\frac{k}{2}$ edges, the adversary immediately stops at this point.) Let us call the selected edges $(v_1,v_2), (v_3,v_4), \ldots, (v_{k-1},v_k)$. The adversary now reveals $\frac{k}{2}+1$ edges (over $\frac{k}{2}+1$ timesteps, of course) so that they together with the selected $\frac{k}{2}$ edges form an augmenting path of length $k+1$ between two vertices from $U$. To be precise, the adversary reveals the edges $(u_1,v_1),(v_2,v_3),\ldots,(v_{k-2},v_{k-1}),(v_k,u_2)$, one by one, for some $u_1,u_2\in U$. These new edges are all included in the adversary's solution. Once this is done, the adversary again selects new $\frac{k}{2}$ edges in the algorithm's current matching and repeat this process. The adversary does not select an edge that was selected before (and a path containing it has been created) again; likewise, the two endpoints $u_1$ and $u_2$ will always be selected ``fresh'' from $U$. The adversary stops when there remain fewer than $\frac{k}{2}$ edges in the matching that were not previously selected. We remark that the fresh choice of $u_1$ and $u_2$ is always possible because $|U|=2\cdot \frac{|F|}{k/2}$. We also note that, during this process, some edges in $F$ that was not added to the matching during the first $\frac{kN}{2}$ timesteps may newly enter the matching, but in any case, there will remain fewer than $\frac{k}{2}$ unselected matching edges when the adversary stops.

At any point of time, if the algorithm wants to remove some edge from the matching, then it needs to honor the commitment that its endpoints will remain matched. Due to the adversary's construction, this enforces the algorithm to choose the entire augmenting path of length $k+1$. Similarly, if the algorithm wants to add to the matching any edge outside $F$, it is again required to choose the entire augmenting path of length $k+1$. Neither is possible because they involve $k+2$ (re)assignments. This proves our earlier claim that the algorithm can only add the edges in $F$ to the matching, and that once an edge enters the matching it will never leave.

Let us now bound $\mathsf{ADV}$. Since we can consider an imaginary adversary that adds the first $\frac{kN}{2}$ edges to its solution and stops immediately after time~$\frac{kN}{2}$, the algorithm must choose at least $(1 - \frac{2}{k + 2} + \epsilon)\cdot\frac{kN}{2}-c$ edges from $F$ in expectation by the end of time $\frac{kN}{2}$, even against the given adversary. Let $\mathsf{ALG}$ be the number of edges that are output by the algorithm at the end.
The graph contains at least $\lfloor \frac{\mathsf{ALG}}{\frac{k}{2}} \rfloor$ paths and hence we have $\mathsf{ADV}= \lfloor \frac{\mathsf{ALG}}{\frac{k}{2}} \rfloor \cdot (\frac{k}{2}+1)$.
For $N > \frac{1}{\epsilon} ( 1 + \frac{2c}{k}) $, we have $\E [\mathsf{ALG}] <(1 - \frac{2}{k + 2} + \epsilon)\cdot \E [\mathsf{ADV}]-c$ from $\E [\mathsf{ALG}] \geq (1 - \frac{2}{k + 2} + \epsilon)\cdot\frac{kN}{2}-c$.
\end{proof}

\newcommand{\exposed}{\mathsf{exp}}
\newcommand{\nil}{\mathsf{nil}}

\section{Weighted Case} \label{sec:analysis-weight}
%%%%%%%%%%%%%%%%%%%%%%%%%%%%%%%%%%%%%%%%%%%%%%%%%%%%%%%%%%%%%%%%
% Analysis
%%%%%%%%%%%%%%%%%%%%%%%%%%%%%%%%%%%%%%%%%%%%%%%%%%%%%%%%%%%%%%%%
\subsection{Analysis} \label{subsec:analysis-weight-comp}

In this section, we present the analysis of the given algorithm under the weighted model. Note that we do not assume triangle inequality; only the nonnegativity is required.

\begingroup
\def\thetheorem{\ref{analysis-weight-comp:main}}
\begin{theorem} [restated]
The given algorithm is a $\frac{1}{2}$-competitive algorithm for the maximum-weight online bipartite left-perfect matching problem with general weights under vertex arrivals with a hard budget $k = 4$ on the number of reassignments.
\end{theorem}
\addtocounter{theorem}{-1}
\endgroup

\subparagraph{Notation and Definitions.}
Let $G = (L \cup R, E)$ be a complete bipartite graph, and $w : E \rightarrow \mathbb{Q}_+$ be the edge weights.
Let $u_t\in L$ denote the vertex that arrives at time $t$ and $n:=|L|$ be the number of timesteps.
The algorithm is given $R$ (and $k = 4$) in advance.
Since the model requires the algorithm to output a left-perfect matching, we have $n \leq |R|$.

Let $M_t$ be the matching that the algorithm maintains at the beginning of time $t$, i.e., one that the algorithm outputs at time $t - 1$.
Let $R^t_\exposed$ be the exposed vertices in $R$ at the beginning of time $t$.
Here, the beginning of time $n + 1$ indicates the end of the entire execution; i.e., $M_{n + 1}$ is the final left-perfect matching the algorithm produces.
Given a matching $M$ in $G$, we slightly abuse the notation to let $M(u)$ denote the vertex with which $u\in L \cup R$ is matched by $M$.

Let $p_t$ denote the \emph{marginal profit} that the algorithm achieves at time $t$, i.e., $p_t := \sum_{e \in M_{t + 1}} w(e) - \sum_{e \in M_t} w(e)$ for $1 \leq t \leq n$.
Fahrbach and Zadimoghaddam~\cite{DBLP:journals/corr/Zadimoghaddam17} uses this marginal profit (or marginal gain) to analyze an online algorithm for the weighted problem under the free disposal model.

Suppose that, at time $t$, the algorithm augments the matching at the path $ \langle u_t, v, u, v^\prime \rangle $.
This implies that $(u, v) \in M_t$ and $v^\prime \in R^t_\exposed$. Since the algorithm chose $v$ over $v^\prime$ to be matched with $u$, we have $w(u, v) \geq w(u, v^\prime)$. Thus, swapping  $(u, v)$ out with $(u, v^\prime)$ incurs the \emph{loss} of  $w(u, v) - w(u, v^\prime)\geq 0$ but $w(u_t, v)$ is large enough to cover this loss.
Let us define the \emph{loss from vertex $v\in R$ at time $t$}, denoted by $\ell_v^t$, to be the loss that would incur if the algorithm decides to take $v$ away from $M_t(v)$ and $M_t(v)$ needs to find a new match: i.e.,
\begin{equation*}
\ell_v^t := \begin{cases}
w(M_t(v), v) -  \max\limits_{v^\prime \in R^t_\exposed} w(M_t(v), v^\prime), & \text{if } v \text{ is matched at the beginning of time } t, \\
0, & \text{otherwise},
\end{cases}
\end{equation*}
where $\max \emptyset := 0$.

The following lemma bounds the marginal profit of each timestep.
Let $M^\star$ be a maximum-weight left-perfect matching in $G$.
\begin{lemma} \label{analysis-weight-comp:lem-gain}
For each timestep $t$, we have $p_t \geq w(u_t, M^\star(u_t)) - \ell_{M^\star(u_t)}^t$.
\end{lemma}
\begin{proof}
Let $v := M^\star(u_t)$ be the vertex to which $u_t$ is matched by the optimal left-perfect matching.
Observe that the right-hand side of the inequality represents the profit of matching $u_t$ to $v$ at time $t$:
if $v$ is exposed at the beginning of time $t$, the algorithm can simply add $( u_t, v )$ to the matching, obtaining the profit of  $w(u_t, v)$.
Otherwise, the algorithm can profit at most $w(u_t, v) - w(M_t(v), v) + \max_{v^\prime \in R^t_\exposed} w(M_t(v), v^\prime)= w(u_t, v) - \ell_{v}^t$ if it swaps out $(M_t(v), v)$, matches $M_t(v)$ somewhere else, and adds $( u_t, v )$ to the matching.
Recall that the algorithm finds the most profitable augmenting path of length at most 3 at each timestep.
\end{proof}

Now we show the monotonicity of the loss from each vertex $v\in R$ over time.

\begin{lemma} \label{analysis-weight-comp:lem-non-dec}
For each $v \in R$, $\ell_v^t$ is non-decreasing over $t$.
\end{lemma}
\begin{proof}
Suppose that $v$ newly becomes matched at time $t$.
Since $v$ is exposed at the beginning of this timestep, we have $\ell^t_v = 0$ by definition.
If $\ell^{t + 1}_v < 0$, this implies there is $v^\prime \in R^{t + 1}_\exposed$ such that $w(M_{t + 1}(v), v^\prime) > w(M_{t + 1}(v), v)$.
However, since $v^\prime$ is also exposed at the beginning of time $t$, the algorithm would not have chosen $v$ over $v'$ to match with $M_{t + 1}(v)$. This shows $\ell^{t + 1}_v \geq 0$.

Now suppose that $v$ was matched at the beginning of time $t$.
Since every vertex will never become exposed once it is matched, we have $R^{t + 1}_\exposed \subseteq R^t_\exposed$.
Therefore, if $v$ does not participate in the augmenting path chosen by the algorithm at time $t$ and hence $M_{t + 1}(v)=M_{t }(v)$, we can easily see that $\ell_v^{t + 1} \geq \ell_v^t$.

The only remaining case is when $(M_t(v), v)$ is on the augmenting path the algorithm adopts at time $t$.
Let $P := \langle u_t, v, M_t(v), v_\mathsf{tail} \rangle$ denote this path, where $v_\mathsf{tail}\in R^{t + 1}_\exposed$ is an exposed vertex in $R$ at the beginning of time $t$.
Observe that $v_\mathsf{tail} \in \argmax_{v^\prime \in R^t_\exposed} w(M_t(v), v^\prime)$ by construction of the algorithm.

Let us assume for now that $R^{t + 1}_\exposed \neq \emptyset$.
Let $v_\mathsf{max} \in \argmax_{v^\prime \in R^{t+1}_\exposed} w(u_t, v^\prime)$.
Note that $v_\mathsf{max}$ is also exposed at the beginning of time $t$, implying that the algorithm could have directly matched $u_t$ with $v_\mathsf{max}$, but chose to use $P$ at time $t$.
We thus have
\begin{equation*}
w(u_t, v) - w(M_t(v), v) + w(M_t(v), v_\mathsf{tail}) \geq w(u_t, v_\mathsf{max}),
\end{equation*}
yielding $ \ell^{t + 1}_v = w(u_t, v) - w(u_t, v_\mathsf{max}) \geq w(M_t(v), v) - w(M_t(v), v_\mathsf{tail}) = \ell^t_v$, where the first equality comes from that $u_t$ is the vertex where $v$ is matched by $M_{t + 1}$.

Observe that $R^{t + 1}_\exposed = \emptyset$ only if $t = |R|$.
In this case, since $u_t$ could have been directly matched with $v_\mathsf{tail}$, we have $w(u_t, v) - w(M_t(v), v) + w(M_t(v), v_\mathsf{tail}) \geq w(u_t, v_\mathsf{tail})\geq 0$. This gives $ \ell^{t + 1}_v = w(u_t, v)  \geq w(M_t(v), v) - w(M_t(v), v_\mathsf{tail}) = \ell^t_v. $
\end{proof}

We are now ready to prove Theorem~\ref{analysis-weight-comp:main}.
\begin{proof} [Proof of Theorem~\ref{analysis-weight-comp:main}]
Let $M := M_{n + 1}$ be the left-perfect matching the algorithm produces at the end.
From Lemma~\ref{analysis-weight-comp:lem-gain}, we have
\begin{equation}\label{eq:analysis-weight-comp:lem-gain:1}
\sum_{t = 1}^n p_{t} + \sum_{t = 1}^n \ell^t_{M^\star(u_t)} \geq \sum_{t = 1}^n w(u_{t}, M^\star(u_{t})) = \sum_{e \in M^\star} w(e).
\end{equation}
Observe that, by the definition of $p_t$, we have
\begin{equation}
\sum_{t = 1}^n p_t = \sum_{e \in M} w(e).
\end{equation}
Moreover, by Lemma~\ref{analysis-weight-comp:lem-non-dec}, we can derive
\begin{equation}\label{eq:analysis-weight-comp:lem-gain:3}
\sum_{t = 1}^n \ell^t_{M^\star(u_t)} \leq \sum_{t = 1}^n \ell^{n + 1}_{M^\star(u_t)} \leq \sum_{e \in M} w(e),
\end{equation}
where the last inequality comes from the fact that $M^\star$ is a left-perfect matching.
Combining \eqref{eq:analysis-weight-comp:lem-gain:1}--\eqref{eq:analysis-weight-comp:lem-gain:3} gives $2\sum_{e \in M} w(e) \geq \sum_{e \in M^\star} w(e)$.
\end{proof}

%%%%%%%%%%%%%%%%%%%%%%%%%%%%%%%%%%%%%%%%%%%%%%%%%%%%%%%%%%%%%%%%
% Lower bound
%%%%%%%%%%%%%%%%%%%%%%%%%%%%%%%%%%%%%%%%%%%%%%%%%%%%%%%%%%%%%%%%
\subsection{Lower Bound} \label{subsec:analysis-weight-hard}
In this section, we present the lower bound under the weighted model.

\begingroup
\def\thetheorem{\ref{analysis-weight-hard:lem-hard}}
\begin{theorem} [restated]
For all constant $\epsilon > 0$, no deterministic algorithm is $(\frac{1}{2} + \epsilon)$-competitive for the maximum-weight online bipartite left-perfect matching problem under vertex arrivals with $k = 4$, even when all edge weights are either $0$ or $1$.
\end{theorem}
\addtocounter{theorem}{-1}
\endgroup

We will present an adaptive offline adversary: recall that oblivious adversaries are as powerful as adaptive offline adversaries against deterministic algorithms (see, e.g.,~\cite{DBLP:journals/algorithmica/Ben-DavidBKTW94}).

\subparagraph{Adversary Description.}
Let $N$ be a sufficiently large number. The adversary will construct a graph with $|L|=|R|=3N$; it reveals the graph at each timestep as we describe below. 
Let $u_t$ denote the vertex that arrives at time $t$. On the way, the adversary will name the vertices in $R$ from $v_1$ to $v_{3N}$, although this does not necessarily relate $v_t$ with $u_t$.

\begin{description}
\item{\textsf{Phase 1.}} At time $t=1,\ldots,N$, the adversary reveals only zero-weight edges, i.e., $w(u_t, v) = 0$ for every vertex $v \in R$.
Let $v_t \in R$ be the vertex that got matched by the algorithm at time $t$.

\item{\textsf{Phase 2.}} At time $t=N + 1,\ldots,2N$, the adversary selects a vertex in $R\setminus\{v_1,\ldots,v_{t-1}\}$ to be named $v_t$: if there exists a vertex in $R\setminus\{v_1,\ldots,v_{t-1}\}$ that is matched by the algorithm at the beginning of time $t$, it selects that vertex as $v_t$; otherwise, it selects an arbitrary vertex from $R\setminus\{v_1,\ldots,v_{t-1}\}$. Then, the adversary reveals the edge weights as follows: $w(u_t,v_{t-N})=w(u_t,v_t)=1$ and $w(u_t,v)=0$ for all other $v\in R\setminus\{v_{t-N},v_t\}$.

\item{\textsf{Phase 3.}} At time $t=2N+1,\ldots,3N$, the adversary examines what the algorithm did on the two weight-1 edges incident with $u_{t-N}$, i.e., $(u_{t - N}, v_{t - 2N})$ and $ (u_{t - N}, v_{t - N}) $.
If one of these two edges is in the matching at the beginning of time $t$, the adversary renames these two vertices so that $(u_{t - N}, v_{t - N})$ is the edge in the matching: $(u_{t - N}, v_{t - 2N})$ is not in the matching.
Then, the adversary reveals the edge weights as follows: $w(u_t, v_{t - N})=1$ and $w(u_t,v)=0$ for all other $v\in R\setminus\{v_{t-N}\}$.
\end{description}

\begin{observation} \label{analysis-weight-hard:obsv-5aug}
Suppose that, at some point of time, an edge $(u,v)$ is not in the matching but both $u$ and $v$ are matched. Then, this edge cannot enter the matching for the rest of the algorithm's execution.
\end{observation}
\begin{proof}
In order to add $(u,v)$ to the matching, we need to remove the two edges in the matching incident with $u$ and $v$. We must honor the commitments we made to the other endpoints of these two edges, implying that adding $(u,v)$ to the matching requires we use an augmenting path of length at least 5. Recall that the budget is $k=4$.
\end{proof}

\begin{lemma} \label{analysis-weight-hard:lem-alg-half}
Against the given adversary, every online algorithm produces a perfect matching whose weight is at most $N + 1$.
\end{lemma}
\begin{proof}
By construction, $v_1,\ldots,v_N$ are all matched by the end of Phase 1. For $t=N+1,\ldots,2N+1$, we can show by induction on $t$ that, at the beginning of time $t$, there can be at most one vertex in $R\setminus\{v_1,\ldots,v_{t-1}\}$ that is matched. This implies that at most one vertex in $\{v_{N+1},\ldots,v_{2N}\}$ is exposed at the end of time $2N$. Let $T:=\{t\mid 2N+1\leq t\leq 3N, v_{t-N}\textrm{ is matched at the end of time }2N\}$ and we have $|T|\geq N-1$.

Now consider an arbitrary timestep $t\in T$. Note that both $v_{t-2N}$ and $v_{t-N}$ are matched at the beginning of time $t$.
We now claim that at most one of the three edges $E_t:=\{ (u_{t - N}, v_{t - 2N}), (u_{t - N}, v_{t - N}), (u_{t}, v_{t - N}) \}$ can be in the matching at the end of time $3N$.

If at least one of $(u_{t-N},v_{t-2N})$ and $(u_{t-N},v_{t-N})$ is in the matching, it must be $(u_{t-N},v_{t-N})$ due to the renaming. Since $(u_{t-N},v_{t-2N})$ is not in the matching at the beginning of time $t$, Observation~\ref{analysis-weight-hard:obsv-5aug} implies that $(u_{t-N},v_{t-2N})$ can never enter the matching for the rest of the algorithm's execution. On the other hand, since $(u_{t - N}, v_{t - N})$ and $(u_{t}, v_{t - N})$ share an endpoint, at most one of them can be in the matching, proving the claim.

Now suppose that neither $(u_{t-N},v_{t-2N})$ nor $(u_{t-N},v_{t-N})$ is in the matching. Observation~\ref{analysis-weight-hard:obsv-5aug} implies that these two edges can never enter the matching, proving the claim again.

Note that $E_t:=\{ (u_{t - N}, v_{t - 2N}), (u_{t - N}, v_{t - N}), (u_{t}, v_{t - N}) \}$ are mutually disjoint for $t=2N + 1,\ldots, 3N$, and these are the only edges with weight 1. On the other hand, for all $t\in T$, at most one edge in $E_t$ can be in the final matching. Since $|T|\geq N-1$, this shows that the perfect matching output by the algorithm at the end of the execution has weight at most $N + 1$.
\end{proof}

\begin{proof}[Proof of Theorem~\ref{analysis-weight-hard:lem-hard}]
Note that the graph constructed by the adversary has matching of weight $2N$: choose $(u_{t - N}, v_{t - 2N})$ and $(u_{t}, v_{t - N})$ for all $t=2N+1,\ldots,3N$. The desired conclusion immediately follows from Lemma~\ref{analysis-weight-hard:lem-alg-half} by choosing a sufficiently large $N$.
\end{proof}

\bibliographystyle{plainurl}
\bibliography{lit}

\begin{thebibliography}{10}

\bibitem{DBLP:conf/soda/AggarwalGKM11}
Gagan Aggarwal, Gagan Goel, Chinmay Karande, and Aranyak Mehta.
\newblock Online vertex-weighted bipartite matching and single-bid budgeted
  allocations.
\newblock In {\em Proceedings of the 22nd Annual {ACM-SIAM} Symposium on
  Discrete Algorithms ({SODA})}, pages 1253--1264. {SIAM}, 2011.

\bibitem{DBLP:conf/mfcs/0001DJ18}
Spyros Angelopoulos, Christoph D{\"{u}}rr, and Shendan Jin.
\newblock Online maximum matching with recourse.
\newblock In {\em 43rd International Symposium on Mathematical Foundations of
  Computer Science ({MFCS})}, volume 117, pages 8:1--8:15. Schloss Dagstuhl -
  Leibniz-Zentrum f{\"{u}}r Informatik, 2018.

\bibitem{DBLP:journals/ipl/AvitabileMP13}
T.~Avitabile, C.~Mathieu, and L.~Parkinson.
\newblock Online constrained optimization with recourse.
\newblock {\em Information Processing Letters}, 113(3):81--86, 2013.

\bibitem{DBLP:journals/algorithmica/BansalBGN14}
Nikhil Bansal, Niv Buchbinder, Anupam Gupta, and Joseph Naor.
\newblock A randomized {$O(\log^2 k)$}-competitive algorithm for metric
  bipartite matching.
\newblock {\em Algorithmica}, 68(2):390--403, 2014.

\bibitem{DBLP:journals/algorithmica/Ben-DavidBKTW94}
Shai Ben{-}David, Allan Borodin, Richard~M. Karp, G{\'{a}}bor Tardos, and Avi
  Wigderson.
\newblock On the power of randomization in on-line algorithms.
\newblock {\em Algorithmica}, 11(1):2--14, 1994.

\bibitem{DBLP:journals/jacm/BernsteinHR19}
Aaron Bernstein, Jacob Holm, and Eva Rotenberg.
\newblock Online bipartite matching with amortized {$O(\log^2 n)$}
  replacements.
\newblock {\em Journal of the {ACM}}, 66(5):37:1--37:23, 2019.

\bibitem{DBLP:conf/innovations/BernsteinKPPS17}
Aaron Bernstein, Tsvi Kopelowitz, Seth Pettie, Ely Porat, and Clifford Stein.
\newblock Simultaneously load balancing for every $p$-norm, with reassignments.
\newblock In {\em 8th Innovations in Theoretical Computer Science Conference
  ({ITCS})}, volume~67, pages 51:1--51:14. Schloss Dagstuhl - Leibniz-Zentrum
  f{\"{u}}r Informatik, 2017.

\bibitem{DBLP:journals/sigact/BirnbaumM08}
Benjamin~E. Birnbaum and Claire Mathieu.
\newblock On-line bipartite matching made simple.
\newblock {\em {SIGACT} News}, 39(1):80--87, 2008.

\bibitem{DBLP:conf/focs/BosekLSZ14}
Bartlomiej Bosek, Dariusz Leniowski, Piotr Sankowski, and Anna Zych.
\newblock Online bipartite matching in offline time.
\newblock In {\em 55th {IEEE} Annual Symposium on Foundations of Computer
  Science ({FOCS})}, pages 384--393. {IEEE} Computer Society, 2014.

\bibitem{DBLP:conf/latin/BosekLSZ18}
Bartlomiej Bosek, Dariusz Leniowski, Piotr Sankowski, and Anna
  Zych{-}Pawlewicz.
\newblock A tight bound for shortest augmenting paths on trees.
\newblock In {\em 13th Latin American Symposium on Theoretical Informatics
  ({LATIN})}, volume 10807, pages 201--216. Springer, 2018.

\bibitem{DBLP:conf/wads/BoyarFKL17}
Joan Boyar, Lene~M. Favrholdt, Michal Kotrb{\v{c}}{\'{\i}}k, and Kim~S. Larsen.
\newblock Relaxing the irrevocability requirement for online graph algorithms.
\newblock In {\em 15th International Symposium on Algorithms and Data
  Structures ({WADS})}, volume 10389, pages 217--228. Springer, 2017.

\bibitem{DBLP:conf/esa/BrubachSSX16}
Brian Brubach, Karthik~Abinav Sankararaman, Aravind Srinivasan, and Pan Xu.
\newblock New algorithms, better bounds, and a novel model for online
  stochastic matching.
\newblock In {\em 24th Annual European Symposium on Algorithms ({ESA})},
  volume~57, pages 24:1--24:16. Schloss Dagstuhl - Leibniz-Zentrum f{\"{u}}r
  Informatik, 2016.

\bibitem{DBLP:conf/infocom/ChaudhuriDKL09}
Kamalika Chaudhuri, Constantinos Daskalakis, Robert~D. Kleinberg, and Henry
  Lin.
\newblock Online bipartite perfect matching with augmentations.
\newblock In {\em 28th {IEEE} International Conference on Computer
  Communications ({INFOCOM})}, pages 1044--1052. {IEEE}, 2009.

\bibitem{DBLP:conf/soda/DevanurJK13}
Nikhil~R. Devanur, Kamal Jain, and Robert~D. Kleinberg.
\newblock Randomized primal-dual analysis of {RANKING} for online bipartite
  matching.
\newblock In {\em Proceedings of the 24th Annual {ACM-SIAM} Symposium on
  Discrete Algorithms ({SODA})}, pages 101--107. {SIAM}, 2013.

\bibitem{DBLP:journals/iandc/EpsteinLSW18}
Leah Epstein, Asaf Levin, Danny Segev, and Oren Weimann.
\newblock Improved bounds for randomized preemptive online matching.
\newblock {\em Information and Computation}, 259(1):31--40, 2018.

\bibitem{DBLP:journals/corr/Zadimoghaddam17}
Matthew Fahrbach and Morteza Zadimoghaddam.
\newblock Online weighted matching: Beating the 1/2 barrier.
\newblock {\em CoRR}, abs/1704.05384, 2019.

\bibitem{DBLP:conf/wine/FeldmanKMMP09}
Jon Feldman, Nitish Korula, Vahab~S. Mirrokni, S.~Muthukrishnan, and Martin
  P{\'{a}}l.
\newblock Online ad assignment with free disposal.
\newblock In {\em 5th International Workshop on Internet and Network Economics
  ({WINE})}, volume 5929, pages 374--385. Springer, 2009.

\bibitem{DBLP:conf/focs/FeldmanMMM09}
Jon Feldman, Aranyak Mehta, Vahab~S. Mirrokni, and S.~Muthukrishnan.
\newblock Online stochastic matching: Beating 1-1/e.
\newblock In {\em 50th Annual {IEEE} Symposium on Foundations of Computer
  Science ({FOCS})}, pages 117--126. {IEEE} Computer Society, 2009.

\bibitem{DBLP:conf/soda/GamlathKS19}
Buddhima Gamlath, Sagar Kale, and Ola Svensson.
\newblock Beating greedy for stochastic bipartite matching.
\newblock In {\em Proceedings of the 30th Annual {ACM-SIAM} Symposium on
  Discrete Algorithms ({SODA})}, pages 2841--2854. {SIAM}, 2019.

\bibitem{DBLP:conf/focs/GamlathKMSW19}
Buddhima Gamlath, Michael Kapralov, Andreas Maggiori, Ola Svensson, and David
  Wajc.
\newblock Online matching with general arrivals.
\newblock In {\em 60th {IEEE} Annual Symposium on Foundations of Computer
  Science ({FOCS})}, pages 26--37. {IEEE} Computer Society, 2019.

\bibitem{DBLP:conf/soda/GoelM08}
Gagan Goel and Aranyak Mehta.
\newblock Online budgeted matching in random input models with applications to
  adwords.
\newblock In {\em Proceedings of the 19th Annual {ACM-SIAM} Symposium on
  Discrete Algorithms ({SODA})}, pages 982--991. {SIAM}, 2008.

\bibitem{DBLP:conf/wads/GroveKKV95}
Edward~F. Grove, Ming{-}Yang Kao, P.~Krishnan, and Jeffrey~Scott Vitter.
\newblock Online perfect matching and mobile computing.
\newblock In {\em 4th International Workshop on Algorithms and Data Structures
  ({WADS})}, volume 955, pages 194--205. Springer, 1995.

\bibitem{DBLP:conf/icalp/GuptaGPW19}
Anupam Gupta, Guru Guruganesh, Binghui Peng, and David Wajc.
\newblock Stochastic online metric matching.
\newblock In {\em 46th International Colloquium on Automata, Languages, and
  Programming ({ICALP})}, volume 132, pages 67:1--67:14. Schloss Dagstuhl -
  Leibniz-Zentrum f{\"{u}}r Informatik, 2019.

\bibitem{DBLP:conf/soda/GuptaKS14}
Anupam Gupta, Amit Kumar, and Cliff Stein.
\newblock Maintaining assignments online: Matching, scheduling, and flows.
\newblock In {\em Proceedings of the 25th Annual {ACM-SIAM} Symposium on
  Discrete Algorithms ({SODA})}, pages 468--479. {SIAM}, 2014.

\bibitem{DBLP:conf/wine/HaeuplerMZ11}
Bernhard Haeupler, Vahab~S. Mirrokni, and Morteza Zadimoghaddam.
\newblock Online stochastic weighted matching: Improved approximation
  algorithms.
\newblock In {\em 7th International Workshop on Internet and Network Economics
  ({WINE})}, volume 7090, pages 170--181. Springer, 2011.

\bibitem{DBLP:journals/siamcomp/HopcroftK73}
John~E. Hopcroft and Richard~M. Karp.
\newblock An $n^{5/2}$ algorithm for maximum matchings in bipartite graphs.
\newblock {\em {SIAM} Journal on Computing}, 2(4):225--231, 1973.

\bibitem{DBLP:journals/corr/abs-1910-03287}
Zhiyi Huang.
\newblock Understanding {Z}adimoghaddam's edge-weighted online matching
  algorithm: {W}eighted case.
\newblock {\em CoRR}, abs/1910.03287, 2019.

\bibitem{DBLP:conf/stoc/0002KTWZZ18}
Zhiyi Huang, Ning Kang, Zhihao~Gavin Tang, Xiaowei Wu, Yuhao Zhang, and Xue
  Zhu.
\newblock How to match when all vertices arrive online.
\newblock In {\em Proceedings of the 50th Annual {ACM} {SIGACT} Symposium on
  Theory of Computing ({STOC})}, pages 17--29. {ACM}, 2018.

\bibitem{DBLP:conf/soda/0002PTTWZ19}
Zhiyi Huang, Binghui Peng, Zhihao~Gavin Tang, Runzhou Tao, Xiaowei Wu, and
  Yuhao Zhang.
\newblock Tight competitive ratios of classic matching algorithms in the fully
  online model.
\newblock In {\em Proceedings of the 30th Annual {ACM-SIAM} Symposium on
  Discrete Algorithms ({SODA})}, pages 2875--2886. {SIAM}, 2019.

\bibitem{DBLP:conf/icalp/0002TWZ18}
Zhiyi Huang, Zhihao~Gavin Tang, Xiaowei Wu, and Yuhao Zhang.
\newblock Online vertex-weighted bipartite matching: Beating 1-1/e with random
  arrivals.
\newblock In {\em 45th International Colloquium on Automata, Languages, and
  Programming ({ICALP})}, volume 107, pages 79:1--79:14. Schloss Dagstuhl -
  Leibniz-Zentrum f{\"{u}}r Informatik, 2018.

\bibitem{DBLP:journals/jal/KalyanasundaramP93}
Bala Kalyanasundaram and Kirk Pruhs.
\newblock Online weighted matching.
\newblock {\em Journal of Algorithms}, 14(3):478--488, 1993.

\bibitem{DBLP:conf/stoc/KarpVV90}
Richard~M. Karp, Umesh~V. Vazirani, and Vijay~V. Vazirani.
\newblock An optimal algorithm for on-line bipartite matching.
\newblock In {\em Proceedings of the 22nd Annual {ACM} Symposium on Theory of
  Computing ({STOC})}, pages 352--358. {ACM}, 1990.

\bibitem{DBLP:conf/esa/KesselheimRTV13}
Thomas Kesselheim, Klaus Radke, Andreas T{\"{o}}nnis, and Berthold
  V{\"{o}}cking.
\newblock An optimal online algorithm for weighted bipartite matching and
  extensions to combinatorial auctions.
\newblock In {\em 21st Annual European Symposium on Algorithms ({ESA})}, volume
  8125, pages 589--600. Springer, 2013.

\bibitem{DBLP:conf/icalp/KhullerMV91}
Samir Khuller, Stephen~G. Mitchell, and Vijay~V. Vazirani.
\newblock On-line algorithms for weighted bipartite matching and stable
  marriages.
\newblock In {\em 18th International Colloquium on Automata, Languages, and
  Programming ({ICALP})}, volume 510, pages 728--738. Springer, 1991.

\bibitem{DBLP:conf/wine/KorulaMZ13}
Nitish Korula, Vahab~S. Mirrokni, and Morteza Zadimoghaddam.
\newblock Bicriteria online matching: Maximizing weight and cardinality.
\newblock In {\em 9th International Conference on Web and Internet Economics
  ({WINE})}, volume 8289, pages 305--318. Springer, 2013.

\bibitem{DBLP:conf/soda/ManshadiGS11}
Vahideh~H. Manshadi, Shayan~Oveis Gharan, and Amin Saberi.
\newblock Online stochastic matching: Online actions based on offline
  statistics.
\newblock In {\em Proceedings of the 22nd Annual {ACM-SIAM} Symposium on
  Discrete Algorithms ({SODA})}, pages 1285--1294. {SIAM}, 2011.

\bibitem{DBLP:journals/jacm/MehtaSVV07}
Aranyak Mehta, Amin Saberi, Umesh~V. Vazirani, and Vijay~V. Vazirani.
\newblock Adwords and generalized online matching.
\newblock {\em Journal of the {ACM}}, 54(5):22, 2007.

\bibitem{DBLP:conf/soda/MehtaWZ15}
Aranyak Mehta, Bo~Waggoner, and Morteza Zadimoghaddam.
\newblock Online stochastic matching with unequal probabilities.
\newblock In {\em Proceedings of the 26th Annual {ACM-SIAM} Symposium on
  Discrete Algorithms ({SODA})}, pages 1388--1404. {SIAM}, 2015.

\bibitem{DBLP:conf/soda/MeyersonNP06}
Adam Meyerson, Akash Nanavati, and Laura~J. Poplawski.
\newblock Randomized online algorithms for minimum metric bipartite matching.
\newblock In {\em Proceedings of the 17th Annual {ACM-SIAM} Symposium on
  Discrete Algorithms ({SODA})}, pages 954--959. {ACM} Press, 2006.

\bibitem{DBLP:conf/approx/Raghvendra16}
Sharath Raghvendra.
\newblock A robust and optimal online algorithm for minimum metric bipartite
  matching.
\newblock In {\em Approximation, Randomization, and Combinatorial Optimization.
  Algorithms and Techniques ({APPROX/RANDOM})}, volume~60, pages 18:1--18:16.
  Schloss Dagstuhl - Leibniz-Zentrum f{\"{u}}r Informatik, 2016.

\bibitem{DBLP:conf/icalp/WangW15}
Yajun Wang and Sam~Chiu{-}wai Wong.
\newblock Two-sided online bipartite matching and vertex cover: Beating the
  greedy algorithm.
\newblock In {\em 42nd International Colloquium on Automata, Languages, and
  Programming ({ICALP})}, volume 9134, pages 1070--1081. Springer, 2015.

\end{thebibliography}

\appendix
\section*{Appendix}
\section{Deferred Proofs} \label{sec:app-defer}
\begin{observation}\label{obs:appendix}
For the maximum-cardinality online bipartite matching problem under vertex arrivals with a hard budget $k > 4$, there exists an instance where the matching output by the given algorithm contains an augmenting path of length no greater than $k - 1$.
\end{observation}
\begin{proof}
Let $\ell := k / 2 + 1$. Consider the following scenario with $\ell+1$ timesteps.
Let $u_t$ denote the vertex that arrives at time $t$, and $M_t$ be the matching that the algorithm outputs at the end of time $t$. We will number the vertices in $R$ as follows: $R:=\{v_1,\ldots,v_{\ell+1}\}$.
\begin{enumerate}
\item For each time $t=1,\ldots, \ell - 1$, edges $(u_t, v_t)$ and $(u_t, v_{t + 1})$ are revealed.
Assume the algorithm matches $u_t$ to $v_{t + 1}$ at each time.
\item At time $\ell$, edge $(u_\ell, v_\ell)$ is revealed.
Observe that the algorithm cannot match $u_\ell$ since the only augmenting path starting with $u_\ell$ with respect to $M_{\ell - 1}$ has length $2 \ell - 1 > k - 1$.
We thus have $M_\ell = M_{\ell - 1}$.
\item At time $\ell+1$, edges $(u_{\ell + 1}, v_{\ell - 1})$ and $(u_{\ell + 1}, v_{\ell + 1})$ are revealed.
Suppose the algorithm chooses an augmenting path $\langle u_{\ell + 1}, v_{\ell - 1}, u_{\ell - 2}, v_{\ell - 2}, u_{\ell - 3}, \cdots, v_2, u_1, v_1 \rangle$ of length $2 \ell - 3 = k - 1$. Note that  $(u_t, v_{t + 1}) \in M_{\ell}$ for $1 \leq t \leq \ell - 1$. 
\end{enumerate}
We can now verify that there is an augmenting path $\langle u_\ell, v_\ell, u_{\ell - 1}, v_{\ell - 1}, u_{\ell + 1}, v_{\ell + 1} \rangle$ with respect to $M_{\ell + 1}$ of length 5.
Recall that $k - 1 \geq 5$ since $k$ is an even number greater than $4$.
\end{proof}

\begingroup
\def\thetheorem{\ref{analysis-cardinality:edge-arrival}}
\begin{theorem} [restated]
The given algorithm is a $(1 - \frac{2}{k + 2})$-competitive algorithm for the maximum-cardinality online matching problem under edge arrivals with a hard budget of $k$ on the number of reassignments.
\end{theorem}
\addtocounter{theorem}{-1}
\endgroup

\begin{proof}
We claim that the length of every augmenting path is at least $k + 1$ at the end of each timestep.
The conclusion then follows from Lemma~\ref{prelim:lem-long-aug}.

Use induction on time $t$.
We define the matching at the end of ``time 0'' as the empty set; then the base case is trivial.
Fix a timestep $t\geq 1$, and let us define $G_{t-1}$ and $G_t$ be the graph at time $t-1$ and $t$, respectively.
Let $M_{t-1}$ denote the matching that is output by the algorithm at the end of time $t-1$, and $(u, v)$ be the edge that arrives at time $t$.
The length of any augmenting path in $G_{t-1}$ with respect to $M_{t-1}$ is at least $k + 1$ from the induction hypothesis.

If the algorithm does nothing at time $t$, this happens because every augmenting path in $G_t$ containing $(u, v)$ with respect to $M_{t - 1}$ has length at least $k + 1$.
Suppose towards contradiction that there is an augmenting path in $G_t$ with respect to $M_{t - 1}$ of length less than $k + 1$. Then this path does not contain $(u, v)$ and is an augmenting path in $G_{t - 1}$, contradicting the induction hypothesis.

The remaining case is when the algorithm uses a shortest augmenting path $P$ to output a new matching $M_{t-1} \triangle P$ at the end of  time $t$.
Let $Q$ be an arbitrary augmenting path with respect to $M_{t-1} \triangle P$.
Let $X$ and $Y$ be the augmenting paths in $G_{t}$ with respect to $M_{t-1}$ satisfying the conditions of Lemma~\ref{keylemma:main}.
Due to Condition~\ref{keylemma:cond-disjoint}, at most one of $X$ and $Y$ contains $(u, v)$.
Suppose without loss of generality that $(u, v) \notin Y$.
This implies that $Y$ is an augmenting path in $G_{t-1}$ with respect to $M_{t-1}$, yielding $|Y| \geq k + 1$.
On the other hand, no matter whether $X$ contains $(u, v)$ or not, we have $|X| \geq |P|$ since $P$ is a shortest augmenting path containing $(u, v)$ whose length is no more than $k - 1$.
By Condition~\ref{keylemma:cond-contained} of Lemma~\ref{keylemma:main}, we have $|P| + |Q| \geq |X| + |Y| \geq |P| + k + 1$, completing the proof.
\end{proof}

\end{document}